\documentclass[sn-mathphys-num]{sn-jnl}

\usepackage{booktabs}
\usepackage{mathrsfs}

\usepackage{graphicx}%
\usepackage{multirow}%
\usepackage{amsmath,amssymb,amsfonts}%
\usepackage{amsthm}%
\usepackage{mathrsfs}%
\usepackage[title]{appendix}%
\usepackage{xcolor}%
\usepackage{textcomp}%
\usepackage{manyfoot}%
\usepackage{booktabs}%
\usepackage{algorithm}%
\usepackage{algorithmicx}%
\usepackage{algpseudocode}%
\usepackage{listings}%

\newtheorem{theorem}{Theorem}
\newtheorem{proposition}[theorem]{Proposition}%
\newtheorem{collary}[theorem]{Collary}%
\newtheorem{lemma}[theorem]{Lemma}%

\newtheorem{remark}{Remark}%

\newtheorem{definition}{Definition}%

\begin{document}
	
	\title[The Radical Solution and Computational Complexity]{The Radical Solution and Computational Complexity}
	
	
	\author*[1]{\fnm{Bojin} \sur{Zheng}}\email{zhengbj@mail.scuec.edu.cn}
	
	\author[2]{\fnm{Weiwu} \sur{Wang}}\email{hiwangww@163.com}
	
	
	\affil*[1]{\orgdiv{Department}, \orgname{South-Central Minzu University}, \orgaddress{\street{182 Minzu Road, Hongshan District}, \city{Wuhan}, \postcode{430074}, \state{Hubei}, \country{China}}}
	
	\affil[2]{\orgdiv{Department of Strategy}, \orgname{China Development Bank}, \orgaddress{\street{18 Fuxing Inner Street, West District}, \city{Beijing}, \postcode{100080}, \state{Beijing}, \country{China}}}
	





\abstract{The radical solution of polynomials with rational coefficients is a famous solved problem. This paper found that it is a $\mathbb{NP}$ problem. Furthermore, this paper found that arbitrary  $ \mathscr{P} \in \mathbb{P}$ shall have a one-way running graph $G$, and have a corresponding $\mathscr{Q} \in \mathbb{NP}$ which have a two-way running graph $G'$, $G$ and $G'$ is isomorphic, i.e., $G'$ is combined by $G$ and its reverse $G^{-1}$.  When $\mathscr{P}$ is an algorithm for solving polynomials,  $G^{-1}$ is the radical formula. According to Galois' Theory, a general radical formula does not exist. Therefore, there exists an $\mathbb{NP}$, which does not have a general, deterministic  and polynomial time-complexity algorithm, i.e., $\mathbb{P} \neq \mathbb{NP}$. Moreover, this paper pointed out that this theorem actually is an impossible trinity. }


\maketitle

\section{Introduction}\label{sec1}

Now the $\mathbb{P}\ vs\ \mathbb{NP}$ problem or the $\mathbb{P} \overset{\text{? }}{=} \mathbb{NP}$ is the most important problem in computer science. It was proposed by Stephen Cook in 1971 when he published his paper \textit{The complexity of theorem-proving procedures} \cite{Cook71thecomplexity}. In 2000, the question was selected as one of the seven Millennium Prize Questions. For decades, a large number of researchers have developed various technologies to try to solve this problem\cite{luo2015en}, but have not yet solved it completely.

Colloquially, $\mathbb{P}$ refers to the set of problems solvable by a Deterministic Turing machine ($DTM$ or $TM$) in polynomial time, and $\mathbb{NP}$ refers to the set of problems verifiable by a Deterministic Turing machine in polynomial time. $\mathbb{P}\ vs\ \mathbb{NP}$ problem or $\mathbb{P} \overset{\text{? }}{=} \mathbb{NP}$ problem means: whether any problem that can be efficiently verified  can be efficiently solved or not. Cook proposed $NP-$completeness \cite{Cook71thecomplexity} and described the most difficult $NP-$complete problems in $\mathbb{NP}$. $NP-$complete problems can be reduced to $3SAT$ problem, that is, if one $NP-$complete problem can be solved efficiently, then all $NP-$complete problems can be solved efficiently. At present, thousands of  $NP-$complete problems have been found in practice, covering dozens of disciplines including biology, chemistry, economy etc. To solve $\mathbb{P} \overset{\text{? }}{=} \mathbb{NP}$ problem has not only important practical significance, but also important scientific and philosophical significance.

The definition of $\mathbb{NP}$ bases upon the Non-Deterministic Turing machine ($NDTM$), which can be regarded as a general, foreseen, polynomial time-complexity algorithm. Because $\mathbb{P} \subseteq \mathbb{NP}$, $\mathbb{P} \neq \mathbb{NP}$ actually means $\mathbb{P} \subset \mathbb{NP}$, or say, there is at least one element in $\mathbb{NP}$ which does not belong to $\mathbb{P}$, i.e., which can not be resolved by a general, unforeseen (or deterministic), polynomial time-complexity algorithm.

Fundamentally, $\mathbb{NP}$ can be solved by enumerating all the potential solutions to verify them, that is, can be solved by exponential time-complexity algorithms. $\mathbb{NP} \subseteq EXP$. Every problem in $\mathbb{NP}$ can be solved by a general, unforeseen (or deterministic), exponential time-complexity algorithm.
 
This paper found that the solution of roots of the polynomials with rational coefficients is an example of $\mathbb{P} \neq \mathbb{NP}$.

The radical roots of the polynomials is a famous solved problem\cite{Birkhoff1994}. Based on the work of Galois, we have known that there is no radical roots for higher order polynomial equations. The aim of this paper is to reduce the $\mathbb{P} \ vs \ \mathbb{NP}$ problem, i.e., the non-existence of $\mathbb{P}$ algorithm, to the non-existence of the solution of radical roots.  

When we have known a root of any polynomial, we can validate it through computing the value of the polynomial. Obviously, the time-complexity is polynomial. This fact illustrates that the root problem is a $\mathbb{NP}$.

According to the relationship between $\mathbb{NP}$ and $EXP$, we can see that an enumerating algorithm can search the whole set of algebraic numbers and find the roots, especially, because the number of roots is the same as the degree of polynomials, the whole set of roots can be enumerated with an extra cost of times of polynomial time-complexity. 

Noticed that there exist some algorithms which can solve some special cases, for examples, $x^n -2=0$; and when the degree is smaller than and equivalent to $4$. When we aim to prove $\mathbb{P} \neq \mathbb{NP}$, actually, it is a need to prove an impossible trinity among generality, determinism, and the polynomial time-complexity. This fact implies a secretive relationship between $\mathbb{P} \neq \mathbb{NP}$ and No Free Lunch Theorems\cite{WolpMacr97}. We will not discuss this issue in this paper. 

In the next step, we can employ Galois' Theory to prove that there is no a general, deterministic algorithm to solve the polynomials with rational coefficients. We need to utilize the reversible property of $\mathbb{P}$. Arbitrary $ \mathscr{P} \in \mathbb{P}$ shall have a one-way running graph $G$, and have a corresponding $\mathscr{Q} \in \mathbb{NP}$ which have a two-way running graph $G'$, $G$ and $G'$ is isomorphic, i.e., $G'$ is combined by $G$ and its reverse $G^{-1}$.  When $\mathscr{P}$ is an algorithm for solving polynomials,  $G^{-1}$ is the radical formula. According to Galois' Theory, a general radical formula does not exist. Therefore, there exists a $\mathbb{NP}$, which does not have a general, deterministic and polynomial time-complexity algorithm, i.e., $\mathbb{P} \neq \mathbb{NP}$.

\section{Background}
\begin{definition}[Turing Machine]
	A Turing Machine $\mathscr{M}$ is a tuple $\left\langle \Sigma, \Gamma, Q, \delta \right\rangle $. $ \Sigma $ is a finite nonempty set with input symbols; $ \Gamma $ is a finite nonempty set, including a blank symbol $ b $ and $ \Sigma $; $ Q $ is a set of possible states, $ q_{0} $ is the initial state, $ q_{accept} $ is an accept state, $ q_{reject} $ is a reject state; $ \delta $ is a transition function, satisfying
	\[ \delta:(Q - \left\lbrace q_{accept}, q_{reject}\right\rbrace ) \times \Gamma \to Q \times \Gamma \times \left\lbrace -1, 1 \right\rbrace  \]
	if $ \delta(q,s)=(q',s',h) $, the interpretation is that, if $M$ is in state $q$ scanning the symbol $s$, then $q'$ is the new state, $s'$ is the symbol printed, and the tape head moves left or right one square depending on whether $h$ is $-1$ or $1$. $ C $ is a configuration of $\mathscr{M}$, $ C=xqy $, $ x,y \in \Gamma^{*} $,$ y $ is not empty, $ q \in Q $. The computation of $\mathscr{M}$ is the unique sequence $ C_1, C_2, \cdots $. If the computation is finite, then the number of steps is the number of configurations.		
\end{definition}

\begin{definition}[Non-Deterministic Turing Machine]
	A Non-Deterministic Turing Machine $\mathscr{N}$ is a tuple $ \left\langle \Sigma, \Gamma, Q,\Delta \right\rangle $. The difference between $ \mathscr{N} $ and $ \mathscr{M} $ lies in the transition function. There are a set of functions $ \Delta = \{ \delta_1 , \delta_2, \cdots, \delta_m\}$ in $ \mathscr{N} $. In any configuration $ C_i $, you can choose between the different functions. The choice is denoted as $ \Delta_{k,j} $, $ j \in \{1,2,\cdots, m\} $. The sequence of choices is a path of computation. If there exits a path $ \Delta_{1,j_1}, \Delta_{2,j_2}, \cdots, \Delta_{k,j_k}  $ leading to the accepting state, $\mathscr{N} $ halts  and accepts. The number of computation $k$ is the number of steps in the shortest path.
\end{definition}

\begin{definition}[The definition of $\mathbb{P}$.]
	
	Assume that $f: \mathbb{N} \xrightarrow[{\ddagger dtm \ddagger}]{} \{0,1\}$, $dtm \in DTM$, and the semantic meaning of $dtm$ is the function $f$, that is, $\ddagger dtm \ddagger=f$, $O(dtm) = O(n^c)$, here, $n$ is the length of input, and $c$ is a constant, $f \in \mathbb{P}$, and $dtm \in \mathbb{P}_f$. Here, $\mathbb{P}_f$ is the set of algorithms to solve the polynomial time-complexity problems $\mathbb{P}$.	
\end{definition}

\begin{definition}[The definition of $\mathbb{NP}$.]

	Assume that $f: \mathbb{N} \xrightarrow[{\ddagger ntm \ddagger}]{} \{0,1\} $, here, $ntm \in NDTM$,  the semantic meaning of $ntm$ is $f$, that is, $\ddagger ntm \ddagger=f$, $O(ntm \rightarrow dtm) = O(n^c)$ , $n$ is the length of input, $c$ is a constant, thus, $f \in \mathbb{NP}$, $ntm \in \mathbb{NP}_f$.		
\end{definition}

\section{Main Results}

\subsection{The Idle Algorithm}
The idle algorithm is an algorithm with idle instructions. 

An algorithm can spend exponential time to execute idle instructions. On one side, because exponential time has been spent, we can say this algorithm has exponential time-complexity. On the other side, because these exponential times of execution of instructions can be eliminated totally, therefore, this algorithm has no time-complexity. 

Of course, the idle  algorithm can also be discussed when spending polynomial time. 

Obviously, the idle algorithm can have polynomial time-complexity, and simultaneously, exponential time-complexity, and $O(1)$ time-complexity.  

It is easy to prove that the idle instructions are feasible to detect. 

If an algorithm has no idle instructions, then this algorithm is called no-idle algorithm.

\subsection{The Zero Algorithm}
The zero algorithm is an algorithm whose code length is zero. 

According to the definitions of $\mathbb{P}$ and $\mathbb{NP}$, the length of $dtm$ is not zero by default.  When length of $dtm$ is zero, we call it zero algorithm.

Because any algorithm can be revised to have the value zero when it halts, we assume that all algorithms will halt when its value is zero. Before it halts, it can output its original value.

If the code length of an algorithm is not zero, then we call this algorithm non-zero algorithm.

\subsection{Running graph}

\begin{definition}[Running graph of $DTM$.]
	Assume that $f: X \xrightarrow[{\ddagger dtm \ddagger}]{} \{0,1\}$, $dtm \in DTM$, $X \subseteq \mathbb{N} $ and the semantic meaning of $dtm$ is the function $f$, that is, $\ddagger dtm \ddagger=f$, let $X' =\{x| f(x) =0 \}$,  $f_k(X')$ is the state of the $k-th$ step of $dtm$. A graph $G= <V, E> $ is called running graph of $dtm$, if $V$ is the set of nodes and $E$ is the set of edges, and  $\forall v \in V ,v= f_k(x) , x \in X', k \in \mathbb{N}$, $ \forall e \in E,e = <f_k(x), f_{k+1}(x)>,  x\in X', k \in \mathbb{N}$.	
\end{definition}

\begin{definition}[Running graph of $Oracle-NDTM$.]
	Assume that a graph $G'= <V, E> $ is called running graph of $ntm$, define $ntm=f^{-1} \circ f : 0  \rightarrow X' \rightarrow 0, ntm \in NDTM$. Here, $G^{-1}=<V, E'>$, $\forall x \in X', \forall e' \in E', e' =<f_{k+1}(x), f_k(x)>$.
	 Here $f: X \xrightarrow[{\ddagger dtm \ddagger}]{} \{0,1\}$, $dtm \in DTM$, $X \subseteq \mathbb{N} $ and the semantic meaning of $dtm$ is the function $f$, that is, $\ddagger dtm \ddagger=f$, $X' =\{x| f(x) =0 \}$. Here, the $NDTM$ is called $Oracle-NDTM$
\end{definition}

\begin{theorem}\label{TM.reverse}
	Any non-zero no-idle total $DTM$ is reversible and has a corresponding  $Oracle-NDTM$.
\end{theorem}

\begin{proof}
	Any non-zero no-idle  total $DTM$ shall have a running graph $G$ of $DTM$. Because of the one-way execution of $DTM$, $G$ is a tree. According to the definition of running graph of $Oracle-NDTM$, there exists a corresponding running graph $G'$ of $Oracle-NDTM$. $G'$ is combined by $G^{-1}$ and $G$. This $Oracle-NDTM$ will start from $0$ and reverse the paths of $DTM$, finally obtain the subset of $\mathbb{N}$ as the input, and then verifies the solutions.  
\end{proof}

\begin{remark}
	Theorem \ref{TM.reverse} illustrated that all non-zero no-idle total $DTM$ is reversible and will have a $Oracle-NDTM$, that is, every computation (modifying the bit string) in $DTM$, i.e., the edge in the running graph of $DTM$, will have a reversible computation (recovering the bit string ) in $Oracle-NDTM$. The running graph of $DTM$ is isomorphic to the running graph of $Oracle-NDTM$ when treating the one-way edges two-way. The running graph of $NDTM$ is a spindle with the symmetry axis of the input.    
\end{remark}

\begin{collary} \label{col.tm.psize}
	Any non-zero no-idle  total $DTM$ with polynomial size of input is reversible and has a  $Oracle-NDTM$.
\end{collary}

\begin{collary} \label{col.pnp}
	Any non-zero no-idle total $\mathbb{P}$ algorithm is reversible and has a $\mathbb{NP}$ algorithm.
\end{collary}

\subsection{The algorithm or formula of polynomials}

\begin{definition}
	An exact algorithm (or an algorithm) of polynomials can only include addition, subtraction, multiplication, division and power function. 
\end{definition}

\begin{definition}
	A $\mathbb{NP}$ or $NEXP$ algorithm (or a $NDTM$ algorithm) of polynomials is a $NDTM$ which can only include addition, subtraction, multiplication, division, power function and taking the n-th root. 
\end{definition}

As to the other algorithms, we call them approximation algorithms, and we will not discuss them in this paper.

For the polynomials with rational coefficients, it is only include addition, subtraction, multiplication, division and power function, i.e., the running graph will only include these limited operators. However, the reversible $Oracle-NDTM$ should have all reversible operators. Therefore, the corresponding $Oracle-NDTM$ must include the operator of taking the $n$-th root.

According to the reversible theorem\ref{col.pnp} of $\mathbb{P}$ algorithms, when polynomials is $\mathbb{NP}$, if there exists a $\mathbb{P}$ algorithm,  there must exists a reversible $\mathbb{NP}$ algorithm whose operators are reversible. 

\begin{proposition}
	The existence of $\mathbb{P}$ algorithms leads to the existence of $\mathbb{NP}$ algorithms for solving the roots of the polynomials with rational coefficients. Here, the $\mathbb{NP}$ algorithm is the formula of the radical roots.
\end{proposition}

\subsection{The Root of Polynomials and $\mathbb{NP}$ 	}
\begin{theorem}\label{thm.np}
	Solving Roots of polynomials over $\mathbb{Q}$ is a $\mathbb{NP}$.
\end{theorem}
\begin{proof}
	According to the definition of $\mathbb{NP}$, when the set of potential solutions is countable, and the verification of any solution has polynomial time-complexity, the problem is a $\mathbb{NP}$.
	Because all algebraic numbers are countable, the first condition satisfies. If we have obtained the representation of a root $x$, then compute the polynomial function $f(x)$, obviously, the time-complexity is $O(n^2)$ (counting the times of addition and multiplication).  Q.E.D.	
\end{proof}

\begin{remark}
The solution of Radical Roots of polynomials over $\mathbb{Q}$ is a $\mathbb{NP}$.
\end{remark}

\begin{collary}\label{col.ndtm}
	According to the definition of $\mathbb{NP}$, there is a $NDTM$ which can be used to solve the roots of polynomials over $\mathbb{Q}$.
\end{collary}

\begin{theorem}\label{thm.exp}
	 There exists an exponential time-complexity, deterministic, general algorithm can be used to solve the roots of polynomials over $\mathbb{Q}$.
\end{theorem}
\begin{proof}
	Because $\mathbb{NP} \subseteq EXP$, Q.E.D.
\end{proof}

\begin{remark}
	Actually, the $EXP$ can enumerate every element in algebraic number set, if it is a root for a certain polynomials, then return true. Notice that all the algebraic numbers, i.e., the roots of polynomials, are encoded by the polynomials and the count of solutions. That is, when we enumerate the algebraic numbers, we can not obtain any information on the roots, except that they are the solutions of which polynomials. 
\end{remark}

\begin{lemma}\label{lem.sp}
	There exists an deterministic polynomial time-complexity algorithm which can be used to solve the roots of some special polynomials over $\mathbb{Q}$.
\end{lemma}

\begin{proof}
	Let the equation is $x^n-2=0$. The radical solution is $x=\sqrt[n]{2}$. Moreover, when the order of polynomials is less than 5, there exists a deterministic algorithm to obtain a solution.  
\end{proof}

\begin{theorem}\label{thm.p.neq.np}
    There exists no general, deterministic, polynomial time-complexity algorithm to solve the radical roots of polynomials over $\mathbb{Q}$. $\mathbb{P} \neq \mathbb{NP}$. 
\end{theorem}
\begin{proof}
	$G^{-1}$ is the radical formula of polynomials with rational coefficients, when a $\mathbb{P}$ algorithm is applied to solve this polynomials.  
	According to Galois theory, there is no formula (a general algorithm) can solve the radical roots of polynomials over $\mathbb{Q}$.That is, there is no polynomial time-complexity $NDTM$ can solve the radical roots.
	According to  Corollary \ref{col.pnp}, based on the reverse theorem, there is no $DTM$. 
	 According  to Theorem \ref{thm.exp}, there exists a general, deterministic, exponential time-complexity algorithm can solve the roots of quintic or higher order polynomials over $\mathbb{Q}$. Therefore, the time-complexity should be set as polynomial time-complexity. Q.E.D.
\end{proof}

\begin{remark}
	Theorem \ref{thm.p.neq.np} actually is an impossible trinity, i.e., the generality, the accuracy (or the determinism ) and the economy ( or the efficiency, or polynomial time-complexity) of an algorithm.
\end{remark}

	 
\section{Conclusions}
This paper discussed the radical roots of polynomials with rational coefficients. The solution of radical roots is a $\mathbb{NP}$, however, there is no general, deterministic, polynomial time-complexity algorithm to solve the higher order polynomial equations. This is an example of $\mathbb{P} \neq \mathbb{NP}$.

\begin{appendices}
	
\section{The Representation of Roots and $\mathbb{NP}$ Algorithm}\label{secA1}
	
The roots of polynomials with rational coefficients can be represented in various styles. For examples, the radical representation, the exponential function representation and the direct representation. The direct representation uses the rational coefficients and the order number to represent the roots. For example, the roots of $x^2 - 2 = 0$ can be represented as $[1,0,-2;1]$ and $[1,0,-2;2]$, i.e., $-\sqrt{2}$ and $\sqrt{2}$.

Under the direct representation, the $\mathbb{NP}$ and $\mathbb{EXP}$ algorithm can be easily obtained, i.e., enumerating all the algebraic numbers with the degree $n$. Here, the direct representation will not provide any additional information to the enumerating algorithms, but only the existence of roots.

\section{The Reversible Property of Turing Machine}
As we have known, when a Turing Machine is a total function $f$, for any given input $X_0$, the steps of $f_k(X_0)$ form a finite semi-group. That is, when $X_0$ is given, the reversible property of $f$  with the finite elements can be well-defined. 

According to the definition of Turing Machine, the $\delta$ function satisfies 

\[ \delta:(Q - \left\lbrace q_{accept}, q_{reject}\right\rbrace ) \times \Gamma \to Q \times \Gamma \times \left\lbrace -1, 1 \right\rbrace  \].

When $\delta ^{-1}$ can be multi-valued, or under the hypothesis that $\delta ^{-1}$ can foresee the configuration, there exists 

\[ \delta ^ {-1}: Q \times \Gamma \to (Q - \left\lbrace q_{accept}, q_{reject}\right\rbrace ) \times \Gamma \times \left\lbrace -1, 1 \right\rbrace  \].

When we confines Turing Machine to be a total function, every step of computation can be reversible with a certain input.
	
\end{appendices}

\bibliography{ref}
\end{document}